%% file: main.tex
\documentclass[sigconf]{acmart}
\def\BibTeX{{\rm B\kern-.05em{\sc i\kern-.025em b}\kern-.08emT\kern-.1667em\lower.7ex\hbox{E}\kern-.125emX}}
    
\usepackage{nicefrac}
\usepackage{siunitx}
\usepackage{array,framed}
\usepackage{booktabs}
\usepackage{
  color,
  float,
  epsfig,
  wrapfig,
  graphics,
  graphicx,
  subcaption
}
\usepackage{textcomp}
\usepackage{setspace}
\usepackage{latexsym,fancyhdr,url}
\usepackage{enumerate}
\usepackage[noend,ruled,linesnumbered]{algorithm2e}
\usepackage{algpseudocode}
\usepackage{graphics}
\usepackage{xparse} 
\usepackage{xspace}
\usepackage{multirow}
\usepackage{csvsimple}
\usepackage{balance}
\usepackage{threeparttable}

\usepackage{
  tikz,
  pgfplots,
  pgfplotstable
}
\usepackage{hyperref}

\usetikzlibrary{
  shapes.geometric,
  arrows,
  external,
  pgfplots.groupplots,
  matrix
}

\pgfplotsset{compat=1.9}

\usepackage{mathtools,}

\DeclareMathAlphabet{\mathcal}{OMS}{cmsy}{m}{n}

\DeclareGraphicsExtensions{%
    .png,.PNG,%
    .pdf,.PDF,%
    .jpg,.mps,.jpeg,.jbig2,.jb2,.JPG,.JPEG,.JBIG2,.JB2}

\begin{document}
\fancyhead{}
\def\thetitle{Truncated Laplace and Gaussian mechanisms of RDP}
\title{\thetitle}

	\author{Jie Fu, Zhiyu Sun, Haitao Liu, Zhili Chen}
	\affiliation{%
		\institution{Shanghai Key Laboratory of Trustworthy Computing, East China Noraml University, China}
	}
	\email{jie.fu@stu.ecnu.edu.cn, 51265902067@stu.ecnu.edu.cn, 51255902133@stu.ecnu.edu.cn, zhlchen@sei.ecnu.edu.cn}
\date{}

\input{abstract}
\maketitle

\input{intro}    
\input{preliminary}
\input{methodology1}
\input{methodology2}
\input{conclusion}


\appendix

\end{document}

%% file: abstract.tex
\begin{abstract}
The Laplace mechanism and the Gaussian mechanism are primary mechanisms in differential privacy, widely applicable to many scenarios involving numerical data. However, due to the infinite-range random variables they generate, the Laplace and Gaussian mechanisms may return values that are semantically impossible, such as negative numbers. To address this issue, we have designed the truncated Laplace mechanism and Gaussian mechanism. For a given truncation interval $[a, b]$, the truncated Gaussian mechanism ensures the same Renyi Differential Privacy (RDP) as the untruncated mechanism, regardless of the values chosen for the truncation interval $[a, b]$. Similarly, the truncated Laplace mechanism, for specified interval $[a, b]$, maintains the same RDP as the untruncated mechanism. We provide the RDP expressions for each of them. We believe that our study can further enhance the utility of differential privacy in specific applications. 

\end{abstract}

%% file: intro.tex
\section{Introduction}
\label{sec:intro}
Data privacy is a critical concern for data owners during the processes of data collection, storage, and publication, including the dissemination of user data statistics. In recent years, differential privacy, as introduced by~\cite{dwork2006calibrating}, has emerged as a popular privacy framework, thanks to its robust mathematical privacy guarantees. It achieves robust privacy assurances by ensuring that it is nearly impossible to discern whether a specific individual's data is part of the dataset or not based on the disclosed information. In recent years, researchers have been investigating more compact privacy lower bounds and better mechanisms to achieve a balance between utility and privacy.


The classical form of differential privacy is referred to as $\epsilon$-dp. It imposes an upper bound $\epsilon$ on the multiplicative distance between the probability distributions of randomized query outputs for any two neighboring datasets. To achieve $\epsilon$-dp, the standard approach involves adding Laplacian noise to the query outputs. Introduced by Dwork et al.~\cite{dwork2014algorithmic}, approximate differential privacy is denoted as $(\epsilon,\delta)$-dp and the factor $\delta>0$ represents the failure probability that the privacy guarantee does not hold. The standard mechanism for preserving $(\epsilon,\delta)$-dp is the Gaussian mechanism, which adds Gaussian noise to the query outputs. Subsequently, Mironov et al.~\cite{IlyaMironov2017RnyiDP}. introduced Rényi differential privacy (RDP) by observing the Rényi divergence between the distributions of two outputs. Rényi differential privacy and $(\epsilon,\delta)$-dp can be interconverted, and is the most prevalent and widely adopted in both academic literature and practical applications due to it's more compact privacy bounds.

The Laplace mechanism and Gaussian mechanism are commonly used for numerical data, but their applications face two key challenges. Firstly, their outputs often lack consistency. For instance, consider adding Laplace or Gaussian noise to count queries; negative results hold no meaningful interpretation, yet they constitute valid outputs of these mechanisms. Secondly, they have the potential to produce noisy values that deviate significantly from the true values, leading to substantial statistical variance. 

To address the aforementioned two challenges, we propose a truncated differential privacy output mechanism. We selectively release the outputs of Laplace or Gaussian noise addition. More precisely, we assess the noisy output, and if it falls within the interval $[a, b]$, we release that noisy value. Otherwise, we regenerate noise until we obtain a value that falls within the desired interval for release. This mechanism allows for control over output consistency and reduces the statistical variance introduced by noise addition within the appropriate interval $[a, b]$.

We note that truncated Laplace mechanism have been studied~\cite{Geng_Ding_Guo_Kumar_2018,Holohan_Antonatos_Braghin_Aonghusa_2018}, They both argue that the algorithm after truncation of the output does not satisfy the same $\epsilon$-dp as before truncation. However, our study shows that the truncated Laplace mechanism satisfies the same $(\epsilon,\delta)$-DP under specific intervals as before truncation under the RDP. Specifically, our contributions are as follows:
\begin{itemize}
    \item We propose truncated Laplace mechanism and Gaussian mechanism that limit the output after noise addition in the interval $[a,b]$.
    \item We give the RDP loss of the truncated Gaussian mechanism and prove theoretically that the truncated Gaussian mechanism satisfies the same RDP as before truncation, regardless of the value of the truncation interval $[a,b]$.
    \item We give the RDP loss of the truncated Laplace mechanism and prove theoretically that the truncated Laplace mechanism satisfies the same RDP as before truncation, for some particular truncation intervals $[a,b]$.
\end{itemize}

%% file: preliminary.tex
\section{Preliminary Knowledge} \label{sec:prelim}

\subsection{Differential Privacy} \label{sec:prelim-dp}
Differential privacy is a rigorous mathematical framework that formally defines data privacy. It requires that a single entry in the input dataset must not lead to statistically significant changes in the output \cite{CynthiaDwork2006CalibratingNT,CynthiaDwork2011AFF,dwork2014algorithmic} if differential privacy holds.

\begin{definition}
({\bf Differential Privacy~\cite{dwork2014algorithmic}}). The randomized mechanism $A$ provides ($\epsilon$,  $\delta$)-Differential Privacy (DP), if for any two neighboring datasets $D$ and $D'$ that differ in only a single entry, $\forall$S $\subseteq$ Range($A$),
\begin{equation}
{\rm Pr}(A(D) \in S) < e^{\epsilon} \times {\rm Pr}(A(D') \in S) + \delta.
\end{equation}
\end{definition}

Here, $\epsilon > 0$ controls the level of privacy guarantee in the worst case. The smaller $\epsilon$, the stronger the privacy level is. The factor $\delta > 0$ is the failure probability that the property does not hold. In practice, the value of $\delta$ should be negligible~\cite{zhu2017differential,papernot2018scalable}, particularly less than $\frac{1}{|D|}$.


By adding random noise, we can achieve differential privacy for a function $f: \mathcal{X}^n \rightarrow \mathbb{R}^d$ according to Definition 2.1. The $l_k$-sensitivity determines how much noise is needed and is defined as follow.

\begin{definition} ({\bf $\mathbf{l_k}$-Sensitivity\cite{dwork2006calibrating}}) For a function $f: \mathcal{X}^n \rightarrow \mathbb{R}^d$, we define its $l_k$ norm sensitivity (denoted as $\Delta_k f$) over all neighboring datasets $x, x^{'} \in \mathcal{X}^n$ differing in a single sample as
\begin{align}
    \text{sup}_{x, x^{'} \in \mathcal{X}^n} ||f(x) - f(x^{'})||_k \leq \Delta_k f.
\end{align}
\end{definition}

In this paper, we focus on $l_2$ sensitivity, i.e., $|| \cdot ||_2$. Additionally, the following Lemma~\ref{lem:prelim-post} ensures the privacy guarantee of post-processing operations.

\begin{lemma}[{\bf Post-processing~\cite{dwork2014algorithmic}}]\label{lem:prelim-post}
Let $\mathcal{M}$ be a mechanism satisfying $(\epsilon, \delta)$-DP. Let $f$ be a function whose input is the output of $\mathcal{M}$. Then $f(\mathcal{M})$ also satisfies $(\epsilon,\delta)$-DP.
\end{lemma}

\subsection{Rényi Differential Privacy}
Rényi differential privacy (RDP) is a relaxation of $\epsilon$-differential privacy, which is defined on Rényi divergence as follows.

\begin{definition}({\bf Rényi Divergence \cite{van2014renyi}}) Given two probability distributions $P$ and $Q$, the Rényi divergence of order $\alpha > 1$ is: 
\begin{align}
D_{\alpha}(P \| Q)=\frac{1}{\alpha-1} \ln \mathbf{E}_{x \sim Q}\left[\left(\frac{P(x)}{Q(x)}\right)^{\alpha}\right],
\end{align}
where $\mathbf{E}_{x \sim Q}$ denotes the excepted value of $x$ for the distribution $Q$, $P(x)$, and $Q(x)$ denotes the density of $P$ or $Q$ at $x$ respectively.
\end{definition}


\begin{definition}({\bf Rényi Differential Privacy (RDP) \cite{IlyaMironov2017RnyiDP}}) For any neighboring datasets $x, x^\prime \in \mathcal{X}^n$, a randomized mechanism $\mathcal{M}: \mathcal{X}^n \rightarrow \mathbb{R}^{d}$ satisfies $(\alpha, R)$-RDP if
\begin{align}
    D_{\alpha}(\mathcal{M}(x) || \mathcal{M}(x^\prime)) \leq R.
\end{align}
\end{definition}

The following Lemma~\ref{lem:conversion} defines the standard form for converting $(\alpha, R)$-RDP to ($\epsilon$,  $\delta$)-DP.

\begin{lemma}\label{lem:conversion}
({\bf Conversion from RDP to DP~\cite{balle2020hypothesis}}). if a randomized mechanism $f : D \rightarrow \mathbb{R}$  satisfies $(\alpha,R)$-RDP ,then it satisfies$(R+\ln ((\alpha-1) / \alpha)-(\ln \delta+ \ln \alpha) /(\alpha-1), \delta)$-DP for any $0<\delta<1$.
\end{lemma}

The following Definition~\ref{definition:Gaussian mechanism of RDP} provides a formal definition of Gaussian mechanism, and a formal RDP guarantee by it.

\begin{definition}({\bf RDP of Gaussian mechanism~\cite{IlyaMironov2017RnyiDP}})\label{definition:Gaussian mechanism of RDP}
	Assuming $f$ is a real-valued function, and the sensitivity of $f$ is $\mu$, the Gaussian mechanism for approximating $f$ is defined as
	\begin{align} \label{equ:Gaussian}
	\mathbf{G}_{\sigma} f(D)=f(D)+N\left(0, \mu^2\sigma^{2}\right),
	\end{align}
	where $N(0, \mu^2\sigma^{2})$ is normally distributed random variable with standard deviation $\mu\sigma$ and mean $0$. Then the Gaussian mechanism with noise $\mathbf{G}_{\sigma}$ satisfies $(\alpha,\alpha / 2\sigma^2)-RDP$.
	
\end{definition}

The following Definition~\ref{definition:Laplace noise of RDP} provides a formal definition of Laplace mechanism, and a formal RDP guarantee by it.

\begin{definition}({\bf RDP of Laplace mechanism~\cite{IlyaMironov2017RnyiDP}}) \label{definition:Laplace noise of RDP} 
$\Lambda(\mu, \lambda)$ is Laplace distribution with mean $\mu$ and scale $\lambda$. Assuming that $f: \mathcal{D} \mapsto \mathbb{R}$ is a function of sensitivity $\mu$. For any $\alpha \ge 1$ and $\lambda > 0$, the Rényi divergence for Laplace distribution as follow.
\begin{align}
\begin{array}{r}
D_{\alpha}(\Lambda(0, \lambda) \| \Lambda(\mu, \lambda))=\frac{1}{\alpha-1} \log \left\{\frac{\alpha}{2 \alpha-1} \exp \left(\frac{\mu\alpha-\mu}{\lambda}\right)\right. \\
\left.+\frac{\alpha-1}{2 \alpha-1} \exp \left(\frac{-\alpha\mu}{\lambda}\right)\right\} .
\end{array}
\end{align}
\end{definition}


%% file: methodology1.tex
\section{Truncated Gaussian mechanism}
In this Section, we introduce the truncated Gaussian mechanism. and perform a privacy analysis of the truncated Gaussian mechanism from the perspective of RDP.

\subsection{Algorithm Description}
\begin{algorithm}
\caption{Gaussian mechanism with selective release}\label{selective publish of Gaussian}
\KwIn{function $f(\cdot)$, dataset $D$, Gaussian distribution $N$, a given interval $[a,b]$}
\KwOut{$A(D)$ that falls in the interval $[a,b]$ after adding Gaussian noise}
$f^{\prime}(D)=min(max(f(D),0),\mu)$ \;
$A(D)=f^{\prime}(D)+N(0,\mu^2\sigma^2) $\;
\While{ $ A(D) < a$ or $A(D) > b$}
{$A(D)=f^{\prime}(D)+N(0,\mu^2\sigma^2) $\;
}
\Return $A(D)$
\end{algorithm}
Now we describe the algorithm of truncated gaussian mechanism in Algorithm~\ref{selective publish of Gaussian}. First the output of the function $f(\cdot)$ needs to be clipped to $[0,\mu]$ obtain the sensitivity $\mu$, followed by adding Gaussian noise with mean 0 and standard deviation $\mu\sigma$ to the clipped values. Then determine whether the noise-added value is in the interval $[a,b]$, and output it if it is, otherwise re-add Gaussian noise with the same coefficients until the noise-added value satisfies the interval before outputting it.

\begin{figure}[hbt]
  \centering
  \begin{subfigure}{0.49\linewidth}
    \centering
    \includegraphics[width=1.0\linewidth]{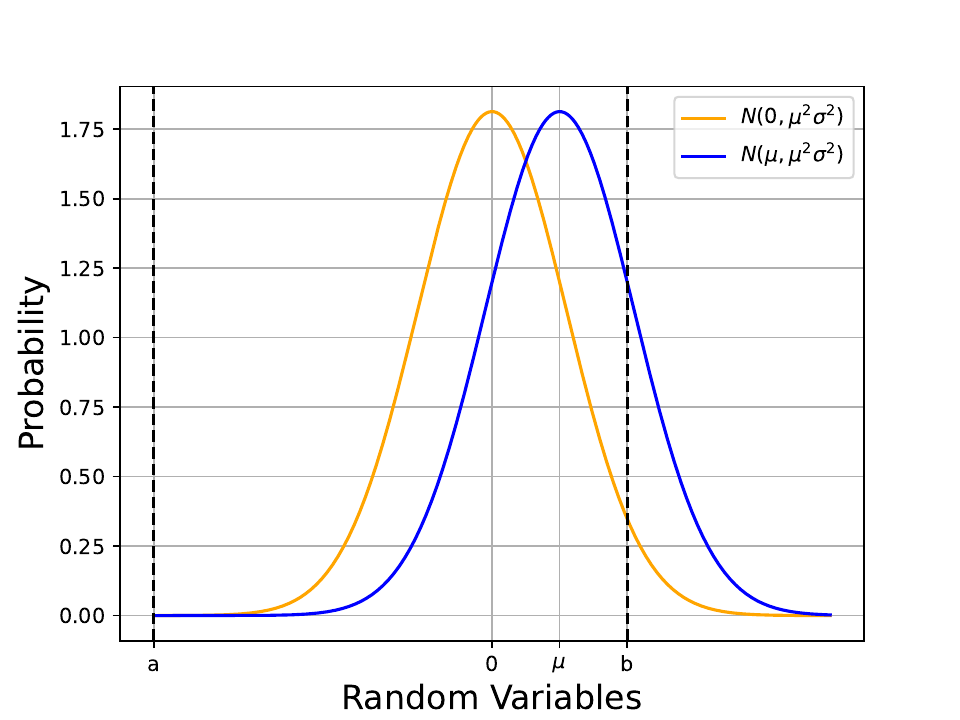}
    \caption{Before truncation.}
    \label{figure: The CDF of Gaussian distribution}
  \end{subfigure}
  \hfill
  \begin{subfigure}{0.49\linewidth}
    \centering
    \includegraphics[width=1.0\linewidth]{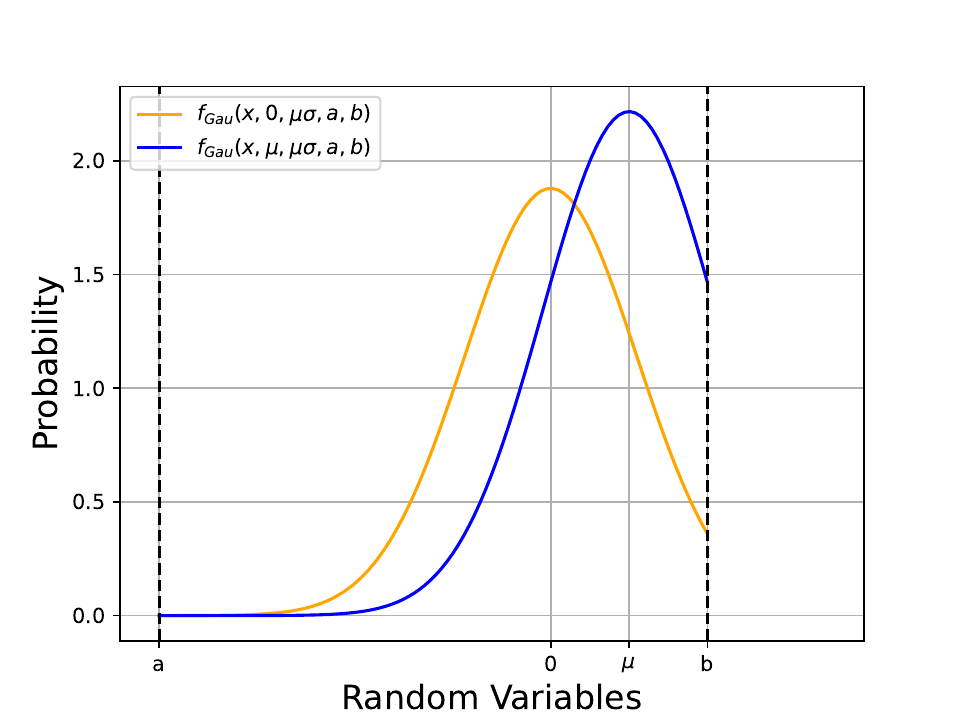}
    \caption{After truncation.}
    \label{figure: Truncated normal distribution}
  \end{subfigure}
  \vspace{5mm}
  \caption{Before and after truncating the normal distribution.}
  \label{figure: Before and after truncating the normal distribution}
\end{figure}

\subsection{Main Result}

\begin{theorem} \label{definition:selective publish of Gaussian}
Algorithm~\ref{selective publish of Gaussian} satisfies $(\alpha,{\alpha}/{2\sigma^2})-RDP$.
\end{theorem}

\begin{proof}
Figure~\ref{figure: Before and after truncating the normal distribution} plots two normal distributions, namely $N(0,\mu^2\sigma^2)$ and $N(\mu,\mu^2\sigma^2)$, as the output probability distributions of function $f$ on two neighboring datasets whose sensitivity is $\mu$. By selective update and selective release in the threshold range $[a,b]$, the probability distribution outside of this interval will accumulate within the interval. As such, the final output distribution is truncated and transformed into a truncated normal distribution, as shown in Figure~\ref{figure: Truncated normal distribution}. As $x$ is assumed to follow a normal distribution, the truncated normal distribution with mean $0$ and mean $\mu$ can be represented as follows:

\begin{align*}
f_{Gau}(x ; 0, \mu\sigma, a, b)=\left\{\begin{array}{ll}
\frac{1}{\mu\sigma \sqrt{2 \pi}} e^{-\frac{x^{2}}{2 \mu^2\sigma^{2}}} \cdot \frac{1}{\Phi\left(\frac{b}{\mu\sigma}\right)-\Phi\left(\frac{a}{\mu\sigma}\right)} & a \leq x \leq b, \\
0 & \text { otherwise. }
\end{array}\right.
\end{align*}

\begin{align*}
f_{Gau}(x ; \mu, \mu\sigma, a, b)=\left\{\begin{array}{ll}
\frac{1}{\mu\sigma \sqrt{2 \pi}} e^{-\frac{(x-\mu)^{2}}{2 \mu^2\sigma^{2}}} \cdot \frac{1}{\Phi\left(\frac{b-\mu}{\mu\sigma}\right)-\Phi\left(\frac{a-\mu}{\mu\sigma}\right)} & a \leq x \leq b, \\
0 & \text { otherwise. }
\end{array}\right.
\end{align*}

Here $\mu\sigma$ is the standard deviation of the original normal distribution, whereas $a$ and $b$ are the lower and upper truncation values, respectively. $\Phi(x)$ denotes the cumulative distribution function of the standard normal distribution.

Then  we substitute the two truncated normal distributions into Rényi divergence \cite{van2014renyi} to calculate the RDP as follows:
\begin{align}
	\begin{split}
		\begin{aligned}
			D_{\alpha}( & f_{Gau}(x ; 0, \mu\sigma, a, b )||f_{Gau}(x ; \mu, \mu\sigma, a, b)) \\
			= & \frac{1}{\alpha-1} \cdot \ln \int_{a}^{b} \frac{[f_{Gau}(x ; 0, \mu\sigma, a, b )]^\alpha}{[f_{Gau}(x ; \mu, \mu\sigma, a, b)]^{\alpha-1}} \mathrm{d} x \\
			= & \frac{1}{\alpha-1}\cdot \ln \{\frac{(\Phi(\frac{b-\mu}{\mu\sigma })-\Phi(\frac{a-\mu}{\mu\sigma } ))^{\alpha-1}  }{(\Phi(\frac{b}{\mu\sigma })-\Phi(\frac{a}{\mu\sigma }))^{\alpha}} \cdot \frac{1}{\mu\sigma \sqrt{2 \pi}} \int_{a}^{b} \exp \left[\left(-x^{2}+\right.\right. \\
			& \left.\left.2(1-\alpha) \mu x-(1-\alpha) \mu^{2}\right) /\left(2 \mu^2\sigma^2\right)\right] \mathrm{d} x \} \\
			= & \frac{1}{\alpha-1} \cdot \{ \frac{\alpha(\alpha-1)}{2 \sigma^2} + \ln [\frac{(\Phi(\frac{b-\mu}{\mu\sigma })-\Phi(\frac{a-\mu}{\mu\sigma } ))^{\alpha-1}  }{(\Phi(\frac{b}{\mu\sigma })-\Phi(\frac{a}{\mu\sigma }))^{\alpha}}  \\
			& \cdot  \int_{\frac{a-(1-\alpha)\mu}{\mu \alpha}}^{\frac{b-(1-\alpha)\mu}{\mu \alpha}} \frac{1}{\sqrt{2\pi}} \cdot exp(-\frac{x^2}{2} )\mathrm{d}(x)]\} \\
			= & \frac{1}{\alpha-1} \cdot \{\frac{\alpha(\alpha-1)}{2\sigma ^2} + \ln [ \frac{(\Phi(\frac{b-\mu}{\mu\sigma })-\Phi(\frac{a-\mu}{\mu\sigma } ))^{\alpha-1}  }{(\Phi(\frac{b}{\mu\sigma })-\Phi(\frac{a}{\mu\sigma }))^{\alpha}} \\
			& \cdot ({\Phi(\frac{b-(1-\alpha)\mu}{\mu\sigma})-\Phi(\frac{a-(1-\alpha)\mu}{\mu\sigma})})]\} \\
			= & \frac{\alpha }{2 \sigma^2} + \frac{1}{\alpha-1} \cdot \ln\{\frac{(\Phi(\frac{b-\mu}{\mu\sigma })-\Phi(\frac{a-\mu}{\mu\sigma } ))^{\alpha-1}  }{(\Phi(\frac{b}{\mu\sigma })-\Phi(\frac{a}{\mu\sigma }))^{\alpha}} \cdot [\Phi(\frac{b-(1-\alpha)\mu}{\mu\sigma})\\
			& -\Phi(\frac{a-(1-\alpha)\mu}{\mu\sigma})]\} ,\\
\end{aligned}
\end{split}
\end{align}

where $\Phi(x)=\frac{1}{\sqrt{2 \pi}} \int_{-\infty}^{x} e^{-\frac{t^{2}}{2}} dt \nonumber$.

Similarly, we can get :
\begin{align}
	\begin{split}
		\begin{aligned}
        D_{\alpha}( & f_{Gau}(x ;\mu, \mu\sigma, a, b )||f_{Gau}(x ; 0, \mu\sigma, a, b)) \\
        = & \frac{1}{\alpha-1} \cdot \ln \int_{a}^{b} \frac{[f_{Gau}(x ; \mu, \mu\sigma, a, b )]^\alpha}{[f_{Gau}(x ; 0, \mu\sigma, a, b)]^{\alpha-1}} \mathrm{d} x \\
        = & \frac{1}{\alpha-1}\cdot \ln \{\frac{(\Phi(\frac{b}{\mu\sigma })-\Phi(\frac{a}{\mu\sigma } ))^{\alpha-1}  }{(\Phi(\frac{b-\mu}{\mu\sigma })-\Phi(\frac{a-\mu}{\mu\sigma }))^{\alpha}} \cdot \frac{1}{\mu\sigma \sqrt{2 \pi}} \int_{a}^{b} \exp \left[\left(-x^{2}+\right.\right. \\
        & \left.\left.2 \alpha \mu x- \alpha \mu^{2}\right) /\left(2 \mu^2\sigma^2\right)\right] \mathrm{d} x \}\\
        = & \frac{1}{\alpha-1} \cdot \{\frac{\alpha(\alpha-1)}{2\sigma ^2} + \ln [ \frac{(\Phi(\frac{b}{\mu\sigma })-\Phi(\frac{a}{\mu\sigma } ))^{\alpha-1}  }{(\Phi(\frac{b-\mu}{\mu\sigma })-\Phi(\frac{a-\mu}{\mu\sigma }))^{\alpha}} \\
        & \cdot ({\Phi(\frac{b-\alpha \mu}{\mu\sigma})-\Phi(\frac{a-\alpha \mu}{\mu\sigma})})]\} \\
		\end{aligned}
	\end{split}
\end{align}


According Theorem~\ref{theorem:Truncated Gaussian mechanism}, we have:
\[
	D_{\alpha}(f_{Gau}(x ; 0, \mu\sigma, a, b )||f_{Gau}(x ; \mu, \mu\sigma, a, b)) \leq \alpha /{2 \sigma^{2}}, and
\]
\[
	D_{\alpha}(f_{Gau}(x ; \mu, \mu\sigma, a, b )||f_{Gau}(x ; 0, \mu\sigma, a, b)) \leq \alpha /{2 \sigma^{2}}.
\]

Therefore, Theorem~\ref{definition:selective publish of Gaussian} is proved.
\end{proof}

\label{sec:Truncated Gaussian mechanism}
\begin{theorem} \label{theorem:Truncated Gaussian mechanism}
If \; $A= [(\Phi(\frac{b-\mu}{\mu\sigma })-\Phi(\frac{a-\mu}{\mu\sigma }))^{\alpha-1}] \cdot [\Phi(\frac{b-(1-\alpha)\mu}{\mu\sigma })-\Phi(\frac{a-(1-\alpha)\mu}{\mu\sigma })] / [(\Phi(\frac{b}{\mu\sigma })-\Phi(\frac{a}{\mu\sigma }))^{\alpha}], B= [(\Phi(\frac{b}{\mu\sigma })-\Phi(\frac{a}{\mu\sigma }))^{\alpha-1}] \cdot [\Phi(\frac{b-\alpha\mu}{\mu\sigma })-\Phi(\frac{a-\alpha\mu}{\mu\sigma })]/[({\Phi(\frac{b-\mu}{\mu\sigma })-\Phi(\frac{a-\mu}{\mu\sigma }))^{\alpha}}]$,\; where\; $\Phi(x)=\frac{1}{\sqrt{2 \pi}} \int_{-\infty}^{x} e^{-\frac{t^{2}}{2}} dt$. $A, B \leq 1$,\; when \;$b\ge a, 
\mu>0,\sigma > 0$\; and\; $\alpha \ge 1$.
\end{theorem}
\begin{proof} We let $c=1/\sigma, p=b/\mu, q=a/\mu$  and $h(x)=\int^{c(p-x)}_{c(q-x)}{\Phi(t)}dt$.

So, proving $A, B \leq 1$ is equivalent to proving the following:
\begin{align}
\begin{split}
\begin{aligned}
         & [h(1)]^{\alpha-1} \cdot h(1-\alpha)\leq {[h(0)]^\alpha},\\
        & [h(0)]^{\alpha-1} \cdot h(\alpha)\leq {[h(1)]^\alpha}. \\
\end{aligned}
\end{split}
\end{align}

Further, to prove as follows:
\begin{align} \label{equ:ln}
\begin{split}
\begin{aligned}
         & \frac{\alpha-1}{\alpha} \cdot \ln[h(1)] + \frac{1}{\alpha} \cdot \ln[h(1-\alpha)] \leq \ln[h(0)],\\
        & \frac{\alpha-1}{\alpha} \cdot \ln[h(0)] + \frac{1}{\alpha} \cdot \ln[h(\alpha)] \leq \ln[h(1)]. \\
\end{aligned}
\end{split}
\end{align}

If $\ln^{\prime\prime}[h(x)] \leq 0$, $\ln[h(x)]$ is a concave function, and according to the properties of concave functions, Inequality 2 is proved.

The second derivative of $\ln{[h(x)]}$ is:
\begin{align}
\begin{split}
\begin{aligned}
         \ln^{\prime\prime}{[h(x)}]
          & = \frac{h(x)h''(x)-[h'(x)]^2}{[h(x)]^2}.\\
\end{aligned}
\end{split}
\end{align}

Because $h(x), [h'(x)]^2, [h(x)]^2>0$, when $h''(x) \leq 0$, $\ln''[h(x)]<0$ is hold.

Further, we can get:
\begin{align}
	\begin{split}
		\begin{aligned}
        h(x)h''(x) = & \frac{1}{2\pi}[\int^{c(p-x)}_{c(q-x)}{e^{-\frac{t^2}{2}}}dt] 
          [c^3(q-x)e^{-\frac{c^2(q-x)^2}{2}}- \\ 
        & c^3(p-x)e^{-\frac{c^2(p-x)^2}{2}}] \\
        \leq &  \frac{1}{2\pi}[\int^{c(p-x)}_{c(q-x)}{\frac{t^2e^{-\frac{t^2}{2}}+e^{-\frac{t^2}{2}}}{t^2}}dt][c^3(q-x)e^{-\frac{c^2(q-x)^2}{2}}- \\
        & c^3(p-x)e^{-\frac{c^2(p-x)^2}{2}}] \\
        = & \frac{1}{2\pi}[\frac{e^{-\frac{c^2(q-x)^2}{2}}}{c(q-x)}-\frac{e^{-\frac{c^2(p-x)^2}{2}}}{c(p-x)}][c^3(q-x)e^{-\frac{c^2(q-x)^2}{2}}- \\
        & c^3(p-x)e^{-\frac{c^2(p-x)^2}{2}}] \\
		\end{aligned}
	\end{split}
\end{align}

When $h''(x)>0$, according definition~\ref{definition:Aczel's Inequality}, we have:
\begin{align}
	\begin{split}
		\begin{aligned}
        h(x)h''(x) \leq & \frac{1}{2\pi}[ce^{-\frac{c^2(q-x)^2}{2}}-ce^{-\frac{c^2(p-x)^2}{2}}]^2=[h'(x)]^2 \\
		\end{aligned}
	\end{split}
\end{align}

So $\ln''[h(x)]<0$ is hold, Theorem~\ref{theorem:Truncated Gaussian mechanism} is proved.

\end{proof}

\begin{definition}({\bf Aczel's Inequality~\cite{tian2014new}}) \label{definition:Aczel's Inequality} Let $a_{i}>0, b_{i}>0(i=1,2, \ldots, n), a_{1}^{2}-\sum_{i=2}^{n} a_{i}^{2}>0, b_{1}^{2}-\sum_{i=2}^{n} b_{i}^{2}>0$. Then:
\begin{align}
\left(a_{1}^{2}-\sum_{i=2}^{n} a_{i}^{2}\right)\left(b_{1}^{2}-\sum_{i=2}^{n} b_{i}^{2}\right) \leq\left(a_{1} b_{1}-\sum_{i=2}^{n} a_{i} b_{i}\right)^{2}
\end{align}
\end{definition}

This proves Algorithm~\ref{selective publish of Gaussian} satisfies the same $(\alpha,{\alpha}/{2\sigma^2})$-RDP as Gaussian mechanism of RDP~\cite{IlyaMironov2017RnyiDP}.

%% file: methodology2.tex
\section{Truncated Laplace Mechanism}
\label{sec:Truncated Laplace Mechanism}

In this Section, we introduce the truncated Laplace mechanism. and perform a privacy analysis of the truncated Laplace mechanism from the perspective of RDP.

\subsection{Algorithm Description}
\begin{algorithm}
\caption{Laplace mechanism with selective release}\label{selective publish of Laplace}
\KwIn{function $f(\cdot)$, dataset $D$, Laplace distribution $Lap$, a given interval $[a,b]$}
\KwOut{$A(D)$ that falls in the interval $[a,b]$ after adding Laplace noise}
$f^{\prime}(D)=min(max(f(D),0),\mu)$ \;
$A(D)=f^{\prime}(D)+Lap(0,\lambda) $\;
\While{ $ A(D) < a$ or $A(D) > b$}
{$A(D)=f^{\prime}(D)+Lap(0,\lambda) $\;
}
\Return $A(D)$
\end{algorithm}
Now we describe the algorithm of truncated Laplace mechanism in Algorithm~\ref{selective publish of Laplace}. First the output of the function $f(\cdot)$ needs to be clipped to $[0,\mu]$ obtain the sensitivity $\mu$, followed by adding Laplace noise with mean 0 and distribution parameters $\lambda$ to the clipped values. Then determine whether the noise-added value is in the interval $[a,b]$, and output it if it is, otherwise re-add Laplace noise with the same coefficients until the noise-added value satisfies the interval before outputting it.

\subsection{Main Result}

\begin{figure}[htb]
  \centering
  \begin{subfigure}{0.49\linewidth}
    \centering
    \includegraphics[width=1.0\linewidth]{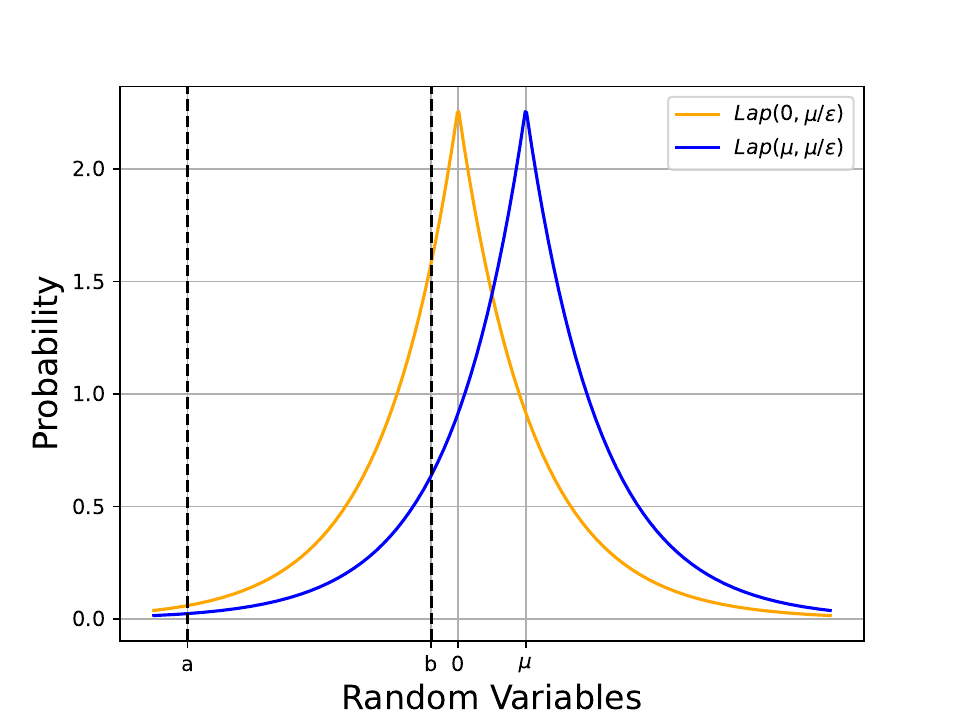}
    \caption{Before truncation.}
    \label{figure: The CDF of Laplace distribution}
  \end{subfigure}
  \hfill
  \begin{subfigure}{0.49\linewidth}
    \centering
    \includegraphics[width=1.0\linewidth]{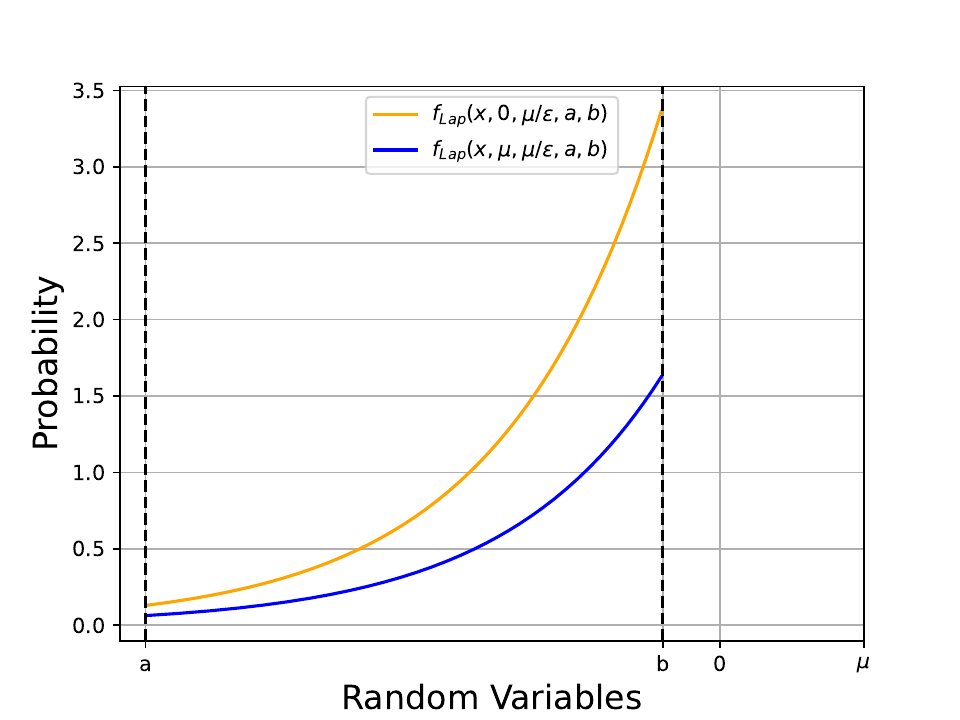}
    \caption{After truncation.}
    \label{figure: Truncated Laplace distribution}
  \end{subfigure}
  \vspace{5mm}
  \caption{Before and after truncating the Laplace distribution.}
  \label{figure: Before and after truncating the Laplace distribution}
\end{figure}

\begin{theorem} \label{definition:selective publish of Laplace}
Algorithm~\ref{selective publish of Laplace} satisfies $(\alpha,\frac{1}{\alpha-1} \log \left\{\frac{\alpha}{2 \alpha-1} \exp \left(\frac{\mu\alpha-\mu}{\lambda}\right)\right. \\
\left.+\frac{\alpha-1}{2 \alpha-1} \exp \left(\frac{-\alpha\mu}{\lambda}\right)\right\})-RDP$ when $a<b<0$ or $\mu<a<b$ or $0<a<b<\mu$.
\end{theorem}

\begin{proof}
    
Figure~\ref{figure: Before and after truncating the Laplace distribution} plots two Laplace distributions, namely $Lap(0,\lambda)$ and $Lap(\mu,\lambda)$, as the output probability distributions of function $f$ on two neighboring datasets whose sensitivity is $\mu$. By selective update and selective release in the threshold range $[a,b]$, the probability distribution outside of this interval will accumulate within the interval. As such, the final output distribution is truncated and transformed into a truncated normal distribution, as shown in Figure~\ref{figure: Truncated Laplace distribution}. As $x$ is assumed to follow a normal distribution, the truncated normal distribution with mean $0$ and mean $\mu$ can be represented as follows:

\begin{align*}
f_{Lap}(x ; 0, \lambda, a, b)=\left\{\begin{array}{ll} \frac{1}{2\lambda} e^{-\frac{|x|}{\lambda}} \cdot \frac{1}{\int_{a}^{b} \frac{1}{2\lambda} e^{-\frac{|t|}{\lambda}} dt} & a \leq x \leq b \\
0 & \text { otherwise } ,
\end{array}\right.
\end{align*}

\begin{align*}
f_{Lap}(x ; \mu, \lambda, a, b)=\left\{\begin{array}{ll} \frac{1}{2\lambda} e^{-\frac{|x-\mu|}{\lambda}} \cdot \frac{1}{\int_{a}^{b} \frac{1}{2\lambda} e^{-\frac{|t-\mu|}{\lambda}} dt} & a \leq x \leq b \\
0 & \text { otherwise } ,\\
\end{array}\right.
\end{align*}

Then  we substitute the two truncated Laplace distributions into Rényi divergence \cite{van2014renyi} to calculate the RDP as follows:

\begin{align}
\begin{split} 
\begin{aligned}
D_{\alpha}( & f_{Lap}(x ; 0, \lambda, a, b )||f_{Lap}(x ; \mu, \lambda, a, b)) \\
= & \frac{1}{\alpha-1} \cdot \ln \int_{a}^{b} \frac{[f_{Lap}(x ; 0, \lambda, a, b )]^\alpha}{[f_{Lap}(x ; \mu, \lambda, a, b)]^{\alpha-1}} \mathrm{d} x \\
= & \frac{1}{\alpha-1} \cdot \log \int_{a}^{b} \frac{e^{[-\alpha|x|+(\alpha-1)|x-\mu|]/\lambda}}{(\int_{a}^{b} e^{-\frac{|x|}{\lambda}} dt)^\alpha \cdot (\int_{a}^{b} e^{-\frac{|x-\mu|}{\lambda}} dt)^{1-\alpha}} \mathrm{d} x \\
= & \frac{1}{\alpha-1} \cdot \log  \frac{\int_{a}^{b} e^{[-\alpha|x|+(\alpha-1)|x-\mu|]/\lambda} \mathrm{d} x}{(\int_{a}^{b} e^{-\frac{|x|}{\lambda}} dt)^\alpha \cdot (\int_{a}^{b} e^{-\frac{|x-\mu|}{\lambda}} dt)^{1-\alpha}} \\
\end{aligned}
\end{split}
\end{align}

Similarly, we can get:
\begin{align}
\begin{split} 
\begin{aligned}
D_{\alpha}( & f_{Lap}(x ; \mu, \lambda, a, b) || f_{Lap}(x ; 0, \lambda, a, b )) \\
= & \frac{1}{\alpha-1} \cdot \log  \frac{\int_{a}^{b} e^{[-(1-\alpha)|x|+\alpha|x|]/\lambda} \mathrm{d} x}{(\int_{a}^{b} e^{-\frac{|x|}{\lambda}} dt)^{1-\alpha} \cdot (\int_{a}^{b} e^{-\frac{|x-\mu|}{\lambda}} dt)^{\alpha}} \\
\end{aligned}
\end{split}
\end{align}

Next, we discuss the range of values in the truncated interval [a,b] in the following three scenarios:

\textbf{I. $a<b<0$}
\begin{align}
\begin{split} 
\begin{aligned}
D_{\alpha}( & f_{Lap}(x ; 0, \lambda, a, b )||f_{Lap}(x ; 1, \lambda, a, b)) \\
= & \frac{1}{\alpha-1} \cdot \log  \frac{\lambda(e^{\frac{b-\mu+\alpha\mu}{\lambda}}-e^{\frac{a-\mu+\alpha\mu}{\lambda}})}{\lambda^\alpha (e^{\frac{b}{\lambda}}-e^{\frac{a}{\lambda}})^\alpha \cdot \lambda^{1-\alpha} (e^{\frac{b-\mu}{\lambda}}-e^{\frac{a-\mu}{\lambda}})^{1-\alpha} } \\
= & \frac{1}{\alpha-1} \cdot \log  \frac{e^{\frac{\alpha\mu-\mu}{\lambda}} \cdot (e^{\frac{b}{\lambda}}-e^{\frac{a}{\lambda}})}{(e^{\frac{b}{\lambda}}-e^{\frac{a}{\lambda}})^\alpha \cdot (e^\frac{-\mu}{\lambda})^{1-\alpha} \cdot (e^{\frac{b}{\lambda}}-e^{\frac{a}{\lambda}})^{1-\alpha} } \\
= & \frac{1}{\alpha-1} \cdot \log  1 \\
\end{aligned}
\end{split}
\end{align}

Similarly, we can get:
\begin{align}
\begin{split} 
\begin{aligned}
D_{\alpha}( & f_{Lap}(x ; \mu, \lambda, a, b) || f_{Lap}(x ; 0, \lambda, a, b )) =  \frac{1}{\alpha-1} \cdot \log  1 \\
\end{aligned}
\end{split}
\end{align}

\textbf{II. $\mu<a<b$}
\begin{align}
\begin{split} 
\begin{aligned}
D_{\alpha}( & f_{Lap}(x ; 0, \lambda, a, b )||f_{Lap}(x ; 1, \lambda, a, b)) \\
= & \frac{1}{\alpha-1} \cdot \log  \frac{\lambda(e^{\frac{-b+\mu-\alpha\mu}{\lambda}}-e^{\frac{-a+\mu-\alpha\mu}{\lambda}})}{\lambda^\alpha (e^{\frac{-b}{\lambda}}-e^{\frac{-a}{\lambda}})^\alpha \cdot \lambda^{1-\alpha} (e^{\frac{\mu-b}{\lambda}}-e^{\frac{\mu-a}{\lambda}})^{1-\alpha} } \\
= & \frac{1}{\alpha-1} \cdot \log  \frac{e^{\frac{\mu-\alpha\mu}{\lambda}} \cdot (e^{\frac{-b}{\lambda}}-e^{\frac{-a}{\lambda}})}{(e^{\frac{-b}{\lambda}}-e^{\frac{-a}{\lambda}})^\alpha \cdot (e^\frac{\mu}{\lambda})^{1-\alpha} \cdot (e^{\frac{-b}{\lambda}}-e^{\frac{-a}{\lambda}})^{1-\alpha} } \\
= & \frac{1}{\alpha-1} \cdot \log  1 \\
\end{aligned}
\end{split}
\end{align}

Similarly, we can get:
\begin{align}
\begin{split} 
\begin{aligned}
D_{\alpha}( & f_{Lap}(x ; \mu, \lambda, a, b) || f_{Lap}(x ; 0, \lambda, a, b )) =  \frac{1}{\alpha-1} \cdot \log  1 \\
\end{aligned}
\end{split}
\end{align}



\textbf{III. $0<a<b<\mu$}
\begin{align}
\begin{split} 
\begin{aligned}
D_{\alpha}( & f_{Lap}(x ; 0, \lambda, a, b )||f_{Lap}(x ; \mu, \lambda, a, b)) \\
= & \frac{1}{\alpha-1} \cdot \log  \frac{\frac{\lambda}{1-2\alpha}(e^{\frac{(1-2\alpha)b-\mu+\mu\alpha}{\lambda}}-e^{\frac{(1-2\alpha)a-\mu+\mu\alpha}{\lambda}})}{\lambda^\alpha (e^{\frac{-a}{\lambda}}-e^{\frac{-b}{\lambda}})^\alpha \cdot \lambda^{1-\alpha} (e^{\frac{b-\mu}{\lambda}}-e^{\frac{a-\mu}{\lambda}})^{1-\alpha} } \\
= & \frac{1}{\alpha-1} \cdot \log  \frac{e^{\frac{\mu\alpha-\mu}{\lambda}}\cdot (e^{\frac{(1-2\alpha)b}{\lambda}}-e^{\frac{(1-2\alpha)a}{\lambda}})}{(1-2\alpha) \cdot (e^{\frac{-a}{\lambda}}-e^{\frac{-b}{\lambda}})^\alpha \cdot  (e^{\frac{-\mu}{\lambda}})^{1-\alpha} \cdot (e^{\frac{b}{\lambda}}-e^{\frac{a}{\lambda}})^{1-\alpha} } \\
= & \frac{1}{\alpha-1} \cdot \log  \frac{ (e^{\frac{(1-2\alpha)b}{\lambda}}-e^{\frac{(1-2\alpha)a}{\lambda}})}{(1-2\alpha) \cdot (e^{\frac{-a}{\lambda}}-e^{\frac{-b}{\lambda}})^\alpha \cdot (e^{\frac{b}{\lambda}}-e^{\frac{a}{\lambda}})^{1-\alpha} } \\
= & \frac{1}{\alpha-1} \cdot \log  \frac{(e^{-\frac{a}{\lambda}})^{(2 \alpha-1)}-(e^{-\frac{b}{\lambda}})^{(2 \alpha-1)}}{(1-2\alpha) \cdot (e^{-\frac{a}{\lambda}}-e^{-\frac{b}{\lambda}})} \cdot e^{-\frac{(a+b)(1-\alpha)}{\lambda}} \\
\end{aligned}
\end{split}
\end{align}
We let $x_2=e^{-\frac{a}{\lambda}}$ and $x_1=e^{-\frac{b}{\lambda}}$, we can get:

\begin{align}
\begin{split} 
\begin{aligned}
D_{\alpha}( & f_{Lap}(x ; 0, \lambda, a, b )||f_{Lap}(x ; \mu, \lambda, a, b)) \\
= &\frac{1}{\alpha-1} \cdot \log \frac{1}{1-2\alpha} \cdot \frac{x_{2}^{(2 \alpha-1)}-x_{1}^{(2 \alpha-1)}}{x_{2}-x_{1}} \cdot\left(x_{2} x_{1}\right)^{(1-\alpha)} \\
= & \frac{1}{\alpha-1} \cdot \log  \frac{1}{1-2\alpha} \cdot   \left(\sum_{i=1}^{2 \alpha-1} x_{2}^{(2 \alpha-1-i)} x_{1}^{(i-1)}\right)\left(x_{2} x_{1}\right)^{(1-\alpha)} \\
= & \frac{1}{\alpha-1} \cdot \log  \frac{1}{1-2\alpha} \cdot  \sum_{i=1}^{2 \alpha-1} (\frac{x_1}{x_2})^{i-\alpha} \\
= & \frac{1}{\alpha-1} \cdot \log  \frac{1}{1-2\alpha} \cdot  \frac{e^{-\frac{(1-a)(b-a)}{\lambda}}-e^{-\frac{\alpha(b-a)}{\lambda}}}{1-e^{\frac{-b-a}{\lambda}}}
\end{aligned}
\end{split}
\end{align}

Similarly, we can get:
\begin{align}
\begin{split} 
\begin{aligned}
D_{\alpha}( & f_{Lap}(x ; \mu, \lambda, a, b) || f_{Lap}(x ; 0, \lambda, a, b )) \\
= & \frac{1}{\alpha-1} \cdot \log  \frac{1}{1-2\alpha} \cdot  \frac{e^{-\frac{(1-a)(b-a)}{\lambda}}-e^{-\frac{\alpha(b-a)}{\lambda}}}{1-e^{\frac{-b-a}{\lambda}}}
\end{aligned}
\end{split}
\end{align}

According Theorem~\ref{theorem:Truncated lap mechanism} and Theorem~\ref{theorem:Truncated lap mechanism2}, Theorem~\ref{definition:selective publish of Laplace} is proved.

\end{proof}


\begin{theorem} \label{theorem:Truncated lap mechanism}
 $ 1 \leq \frac{\alpha}{2 \alpha-1} \exp \left(\frac{\mu\alpha-\mu}{\lambda}\right)+\frac{\alpha-1}{2 \alpha-1} \exp \left(\frac{-\mu\alpha}{\lambda}\right)$,\; when 
$\lambda>0$ and $\alpha \ge 1$.
\end{theorem}
\begin{proof}
Because $(e^x)^{\prime\prime} \ge 0$, $e^x$ is a convex function. According Jensen inequality, we have:
\begin{align} \label{equ:convex}
\begin{split}
\begin{aligned}
         & e^{a_1 x_1 + a_2 x_2} \leq a_1 e^{x_1} + a_2 e^{x_2}.\\
\end{aligned}
\end{split}
\end{align}
So the Theorem~\ref{theorem:Truncated lap mechanism} is proved.
\end{proof}

\begin{theorem} \label{theorem:Truncated lap mechanism2}
 $  \frac{\exp{[-(1-\alpha)(b-a)/\lambda]}-\exp{[-\alpha(b-a)/\lambda]}}{1-\exp{[(-b-a)/\lambda]}} \leq \alpha \exp \left(\frac{\alpha-1}{\lambda}\right)+(\alpha-1) \exp \left(\frac{-\alpha}{\lambda}\right)$, when 
$\lambda>0$, $\alpha \ge 1$ and $b>a>\mu$.
\end{theorem}
\begin{proof}
We let $t_0=e^{-\frac{b-a}{\lambda}}$ and $t_1=e^{-\frac{\mu}{\lambda}}$, so the original inequality is transformed as follows:
\begin{align} 
\begin{split}
\begin{aligned}
         & \frac{t_{0}^{(1-\alpha)}-t_{0}^{\alpha}}{1-t_{0}} \leq (\alpha-1) t_{1}^{\alpha}+\alpha t_{1}^{(1-\alpha)}. \\
\end{aligned}
\end{split}
\end{align}

We let $g(t)=(\alpha-1) t^{\alpha}+\alpha t^{(1-\alpha)}$, and we can get:
\begin{align} 
\begin{split}
\begin{aligned}
         g^{\prime}(t) & =\alpha(\alpha-1)\left[t^{(\alpha-1)}-t^{(-\alpha)}\right] \\
         & =\alpha(\alpha-1) \frac{\left(t^{2 \alpha-1}-1\right)}{t^{\alpha}}. \\
\end{aligned}
\end{split}
\end{align}
Because $\alpha \ge 1$, $g^{\prime}(t) \leq 0$ when $0<t<1$. As $0< t_1<t_0<1$, $g(t_0) \leq g(t_1)$. So the original inequality is transformed as follows:
\begin{align} 
\begin{split}
\begin{aligned}
         & (t_{0}^{(1-\alpha)}-t_{0}^{\alpha})/({1-t_{0}}) \leq (\alpha-1) t_{0}^{\alpha}+\alpha t_{0}^{(1-\alpha)}. \\
        & \rightarrow  {\left[(\alpha-1) t_{0}^{\alpha}+\alpha t_{0}^{(1-\alpha)}\right]\left(1-t_{0}\right)+t_{0}^{\alpha}-t_{0}^{(1-\alpha)} \geqslant 0 } \\
       & \rightarrow  \alpha t_{0}^{\alpha}+(\alpha-1) t_{0}^{(1-\alpha)}+(1-\alpha) t_{0}^{(\alpha+1)}-\alpha t_{0}^{(2-\alpha)} \geqslant 0 \\
      & \rightarrow  \alpha\left(t_{0}^{\alpha}-t_{0}^{(2-\alpha)}\right) \geqslant(1-\alpha)\left[t_{0}^{(1-\alpha)}-t_{0}^{(\alpha+1)}\right] \quad\left(\alpha>1, t_{0}<1\right) \\
     & \rightarrow  (t_{0}^{\alpha}-t_{0}^{(2-\alpha)})/({t_{0}^{(1-\alpha)}-t_{0}^{(\alpha+1)}}) \geqslant \frac{1-\alpha}{\alpha} \\
\end{aligned}
\end{split}
\end{align}
we let $h(x)=t_{0}^{x}$, and turn to prove the following inequality:
\begin{align} 
\begin{split}
\begin{aligned}
         \frac{h(\alpha)-h(2-\alpha)}{h(1-\alpha)-h(\alpha+1)} \geqslant \frac{1-\alpha}{\alpha}
\end{aligned}
\end{split}
\end{align}

Because $\alpha \ge 1$, $1-\alpha < 2-\alpha \leq \alpha < \alpha+1$. And turn to prove the following inequality:
\begin{align} 
\begin{split}
\begin{aligned}
         & \frac{h(\alpha)-h(2-\alpha)}{\alpha-(2-\alpha)} \geqslant \frac{h(1-\alpha)-h(\alpha+1)}{(1-\alpha)-(\alpha+1)} \cdot  \frac{ (1-\alpha)}{\alpha} \cdot \frac{(1-\alpha)-(\alpha+1)}{\alpha-(2-\alpha)} \\
          & \rightarrow \frac{h(\alpha)-h(2-\alpha)}{\alpha-(2-\alpha)} \geqslant \frac{h(1+\alpha)-h(1-\alpha)}{(1+\alpha)-(1-\alpha)} \\
\end{aligned}
\end{split}
\end{align}

$h^{\prime}(x)=t_0^x \ln t_0<0$ and $h^{\prime\prime}(x)=t_0^x \ln^3 t_0 < 0$. So according Theorem~\ref{theorem:Truncated lap mechanism3}, Theorem~\ref{theorem:Truncated lap mechanism2} is proved.
\end{proof}

\begin{theorem} \label{theorem:Truncated lap mechanism3}
For a monotone function $f(\cdot)$, if $f^{\prime}(x) \cdot f^{\prime\prime\prime}(x)>0$. then for $x_3<x_1 \leq x_2<x_4$ and $x_1+x_2=x_3+x_4$, the following equation holds:
\begin{align} 
\begin{split}
\begin{aligned}
\left|\frac{f\left(x_{2}\right)-f\left(x_{1}\right)}{x_{2}-x_{1}}\right|\leq \left|\frac{f\left(x_{3}\right)-f\left(x_{4}\right)}{x_{3}-x_{4}}\right|
\end{aligned}
\end{split}
\end{align}
\end{theorem}

\begin{proof}
Let $x=x_4-x_2=x_1-x_3$, so $0<x<\frac{x_4-x_3}{2}$. And let $g(x)=\frac{f(x_4-x)-f(x_3+x)}{x_4-x_3-2x}$, we have:
\begin{align} 
\begin{split}
\begin{aligned}
 g^{\prime}(x)& =\frac{1}{\left(x_{4}-x_{3}-2 x\right)^{2}}[\left(\frac{d f\left(x_{4}-x\right)}{d x}-\frac{d f\left(x_{3}+x\right)}{d x}\right)\left(x_{4}-x_{3}-2 x\right) \\
 & +2\left(f\left(x_{4}-x\right)-f\left(x_{3}+x\right)\right)] \\
 & = \frac{1}{\left(x_{4}-x_{3}-2 x\right)^{2}}[-\left(f^{\prime}\left(x_{1}\right)+f^{\prime}\left(x_{2}\right)\right)\left(x_{2}-x_{1}\right) \\
 & +2\left(f\left(x_{2}\right)-f\left(x_{1}\right)\right)] \\
\end{aligned}
\end{split}
\end{align}
Nest, we let $x_2-x_1=s>0$ and $h(s)=-s\left(f^{\prime}\left(x_{1}+s\right)+f^{\prime}\left(x_{1}\right)\right)+2\left(f\left(x_{1}+s\right)-f\left(x_{1}\right)\right)$, we have:
\begin{align} 
\begin{split}
\begin{aligned}
h^{\prime}(s) & =-\left(f^{\prime}\left(x_{1}+s\right)+f^{\prime}\left(x_{1}\right)\right)-s f^{\prime \prime}\left(x_{1}+s\right)+2 f^{\prime}\left(x_{1}+s\right) \\
& =-s f^{\prime \prime}\left(x_{1}+s\right)+f^{\prime}\left(x_{1}+s\right)-f^{\prime}\left(x_{1}\right)
\end{aligned}
\end{split}
\end{align}
and $h^{\prime \prime}(s)=-s f^{\prime \prime \prime}\left(x_{1}+s\right)$.

So, when $f^{\prime}(x)<0$ and $f^{\prime\prime\prime}(x)<0$, $h^{\prime \prime}(s)>0, h^{\prime}(s)>h^{\prime}(0)=0, h(s)>h(0)=0$ and $g^{\prime}(x) > 0$. So we have:
\begin{align} 
\begin{split}
\begin{aligned}
0 > g(x) &=\frac{f\left(x_{4}-x\right)-f\left(x_{3}+x\right)}{x_{4}-x_{3}-2 x} \\
&=\frac{f\left(x_{2}\right)-f\left(x_{1}\right)}{x_{2}-x_{1}} \\
&>g(0) > \frac{f\left(x_{4}\right)-f\left(x_{3}\right)}{x_{4}-x_{3}}
\end{aligned}
\end{split}
\end{align}

Inversely, when $f^{\prime}(x)>0$ and $f^{\prime\prime\prime}(x)>0$, $h^{\prime \prime}(s)<0, h^{\prime}(s)>h^{\prime}(0)=0, h(s)<h(0)=0$ and $g^{\prime}(x) <0$. So we have:
\begin{align} 
\begin{split}
\begin{aligned}
0 < g(x) &=\frac{f\left(x_{4}-x\right)-f\left(x_{3}+x\right)}{x_{4}-x_{3}-2 x} \\
&=\frac{f\left(x_{2}\right)-f\left(x_{1}\right)}{x_{2}-x_{1}} \\
&>g(0) < \frac{f\left(x_{4}\right)-f\left(x_{3}\right)}{x_{4}-x_{3}}
\end{aligned}
\end{split}
\end{align}

So, The Theorem~\ref{theorem:Truncated lap mechanism3} is proved. 
\end{proof}

%% file: conclusion.tex
\section{Conclusion}
\label{sec:conclusion}
We have introduced the truncated Laplace and Gaussian mechanisms. For a given truncation interval $[a, b]$, the truncated Gaussian mechanism ensures the same Renyi Differential Privacy (RDP) as the untruncated mechanism, regardless of the values chosen for the truncation interval $[a, b]$. Similarly, the truncated Laplace mechanism, for specified interval $[a, b]$, maintains the same RDP as the untruncated mechanism. We provide the RDP expressions for each of them. In the future, we will further explore the applications of these two mechanisms in specific scenarios.